\documentclass{article}

\usepackage[reset,a4paper,vmargin=3truecm,hmargin=2truecm]{geometry}
\usepackage{amsmath,amsthm,amssymb,amscd,epic}
\usepackage{graphicx,bbm,amsthm,graphicx}
\usepackage{mathtools}
\usepackage{leftidx}
\usepackage{hyperref}
\usepackage{cool}
\usepackage{bm}
\usepackage{bbold}
\usepackage{subfig}

\theoremstyle{plain}
\newtheorem{thm}{Theorem}[section]
\newtheorem{lem}[thm]{Lemma}
\newtheorem{prop}[thm]{Proposition}

\newtheorem{cor}[thm]{Corollary}

\theoremstyle{definition}
\newtheorem{dfn}{Definition}[section]
\newtheorem{ex}{Example}[section]

\theoremstyle{remark}
\newtheorem{rem}{Remark}[section]

\newcommand{\HRabi}[1]{H_{#1}} 

\newcommand{\N}{\mathbb{N}} 
\newcommand{\Z}{\mathbb{Z}} 
\newcommand{\Q}{\mathbb{Q}} 
\newcommand{\R}{\mathbb{R}} 
\newcommand{\C}{\mathbb{C}} 

\DeclareMathOperator{\Spec}{Spec}

\title{Remarks on the hidden symmetry of the asymmetric quantum Rabi model}
\author{Cid Reyes-Bustos, Daniel Braak and Masato Wakayama}

\begin{document}

\maketitle

\begin{abstract}  
  The symmetric quantum Rabi model (QRM) is integrable due to a discrete $\mathbb{Z}_2$-symmetry of the Hamiltonian.
  This symmetry is generated by a known involution operator, measuring the parity of the eigenfunctions. An experimentally
  relevant modification of the QRM, the asymmetric (or biased) quantum Rabi model (AQRM) is no longer invariant under this operator,
  but shows nevertheless characteristic degeneracies of its spectrum for half-integer values of $\epsilon$, the parameter governing the
  asymmetry. 
  In an interesting recent work (arXiv:2010.02496), an operator has been identified which commutes with the Hamiltonian
  $\HRabi{\epsilon}$ of the asymmetric quantum Rabi model for $\epsilon=\frac\ell2 (\ell\in \Z)$ and appears to be the analogue of the parity in the
  symmetric case. We prove several important properties of this operator, notably, that it is algebraically
  independent of the Hamiltonian $\HRabi{\epsilon}$ and that it essentially generates the commutant of $\HRabi{\epsilon}$.
  Then, the expected $\mathbb{Z}_2$-symmetry manifests the fact that the commuting operator can be captured in the two-fold cover of
  the algebra generated by $\HRabi{\epsilon}$, that is, the polynomial ring in $\HRabi{\epsilon}$.

\end{abstract}


\section{Introduction} 

One of the most simple and fundamental models in quantum optics is the Jaynes-Cummings model with Hamiltonian
\begin{equation}
 \label{eq:JC}
  H_{\textrm{JC}} = a^{\dag}a + \Delta \sigma_z + g(a\sigma_+ + a^\dag\sigma_-),
\end{equation}
acting in $\mathcal{H}=L^2(\R)\otimes\C^2$,
with Pauli matrices $\sigma_z$ and $\sigma_{\pm}=\frac12(\sigma_x\pm i\sigma_y)$.
The $a$ ($a^\dag$) are annihilation (creation) operators of a single bosonic mode. Energies are measured in units of the mode frequency ($\hbar=1$). The Hamiltonian \eqref{eq:JC} is invariant under transformations
$\hat{U}(\phi)=\exp i\phi(a^\dag a +\sigma_+\sigma_-)$. Because the spectrum of the generator
$\hat{C}=a^\dag a +\sigma_+\sigma_-$ is $\N_0$, the continuous real parameter $\phi$ may be restricted to the interval $[-\pi,\pi]$, the corresponding symmetry group is thus $U(1)$.

The Jaynes-Cummings model is actually the rotating wave approximation of the quantum Rabi model
\begin{equation}
  \label{eq:Rabi0}
  H_{\textrm{R}} = a^{\dag}a + \Delta \sigma_z + g (a + a^{\dag}) \sigma_x,
\end{equation}
which is still invariant under $\hat{U}(\pi)=-\mathcal{P}\sigma_z$ with $\mathcal{P} = \exp(i\pi a^{\dag} a)$. We have $\hat{U}^2=I$, where $I$ is the identity in $\mathcal{H}$. $\hat{U}(\pi)=-J_0$ is thus an involution and corresponds to the symmetry group $\Z_2$ of \eqref{eq:Rabi0} \cite{B2019}.

While it is pretty obvious that the strong continuous symmetry of the Jaynes-Cummings model is sufficient to solve this model analytically (the Hilbert space separates into the direct sum of two-dimensional subspaces, dynamically invariant under $H_{\textrm{JC}}$ and labeled by an integer quantum number, the eigenvalue of $\hat{C}$), it is not clear whether the same is true for the weak $\Z_2$-symmetry of $H_{\textrm{R}}$. It was argued in \cite{B2013MfI} that this is indeed the case and $H_{\textrm{R}}$ is integrable with respect to the criterion for quantum integrability proposed in \cite{B2011}. Here, we have only two dynamically invariant subspaces $\mathcal{H}_\pm$, labeled by the eigenvalues $\pm1$ of $J_0$, each infinite dimensional, and there is an operator $S_0$, independent from $g$ and $\Delta$, which block-diagonalizes $H_{\textrm{R}}$ as 
\[
  S_0^{-1} H_{\textrm{R}} S_0 =
  \begin{bmatrix}
    H_{+} & 0 \\
    0 & H_{-}
  \end{bmatrix}.
  \]
The spectral graph of $H_{\textrm{R}}$ as function of the coupling $g$ shows characteristic intersections (spectral degeneracies) at certain values of the parameters $g$ and $\Delta$ if an eigenvalue of $H_+$ coincides with one of $H_-$. The degenerate energies correspond to eigenfunctions of $H_{\textrm{R}}$ which are not eigenfunctions of $J_0$ and thus span a two-dimensional (reducible) representation of the symmetry group $\Z_2$. All possible spectral degeneracies of  $H_{\textrm{R}}$ appearing at special values of $g$ and $\Delta$ are related to this symmetry, because the spectra of $H_\pm$ are always non-degenerate \cite{B2013MfI}.  

A simple and physically relevant  generalization of \eqref{eq:Rabi0} (see \cite{Y2017}) is the asymmetric quantum Rabi model,
\begin{equation}
  \label{eq:Rabi}
  \HRabi{\epsilon} := a^{\dag}a + \Delta \sigma_z + g (a + a^{\dag}) \sigma_x + \epsilon \sigma_x,
\end{equation}
where the term $\epsilon\sigma_x$ does not commute with $J_0$ and thus breaks the symmetry of \eqref{eq:Rabi}. Indeed, there are no longer easily recognizable operators like the parity $J_0$ which commute with $\HRabi{\epsilon}$ for $\epsilon\neq0$ and one would guess that the spectrum is non-degenerate for all values of $g$ and $\Delta$. 

This is indeed the case for $\epsilon \notin \frac12 \Z$. However, in the case of half-integer $\epsilon$, degeneracies were observed by Li and Batchelor in \cite{LB2015JPA} for the case $\epsilon = \frac12$ and later shown to happen in the general case by Kazufumi Kimoto and the authors in \cite{KRW2017}, (see also \cite{W2016JPA}), where the degeneracy spectrum was fully clarified. Based on this fact, Semple and Kollar studied the asymptotic behavior of observables in the asymmetric model \cite{SK2017}. Besides these works, the presence of the degeneracies in the spectrum of $\HRabi{\epsilon}$ for the half-integer case leads to the natural question whether these degeneracies are due to a symmetry of $\HRabi{\epsilon}$ which is not ``apparent'' like the parity symmetry of $\HRabi{0}$ but ``hidden''. In this case, one would have to identify the corresponding symmetry group. An initial guess is again $\Z_2$, because the pattern of the degeneracies appearing at half-integer $\epsilon$ are the same as for the symmetric case (two intersecting ladders in the spectral graph as function of $g$ or $\Delta$ \cite{B2019}).

The search for the hidden symmetry was initiated by Ashhab in \cite{A2020}, who showed that the operator commuting with $\HRabi{\epsilon}$, if it exists, must depend on the parameters $g$ and $\Delta$, in contrast to $J_0$ which leads to invariant subspaces $\mathcal{H}_\pm$ which are independent from $g$ and $\Delta$.

In \cite{GD2013}, Gardas and Dajka, using a method based on the Banach fixed point theorem, studied the case of
$\epsilon \notin \frac12 \Z$ with some restriction to the parameters. In particular, when $\epsilon \notin \frac12 \Z$ and $\Delta/\epsilon > \frac{\pi}{2}$, there
is a self-adjoint involution operator $\mathbb{J}_{\epsilon}$ such that
\[
  [\HRabi{\epsilon},\mathbb{J}_{\epsilon}]=0.
\]
It is important to note that the operator  $\mathbb{J}_{\epsilon}$ is given in terms of the solution $X_0$ of a certain Riccati equation
with operator coefficients. The authors also show that by using the solution $X_0$, one can define an operator $S_{\epsilon}$ such that
\[
  S_\epsilon^{-1} \HRabi{\epsilon} S_\epsilon =
  \begin{bmatrix}
    (a^\dag+g)(a+g) +\epsilon -g^2 + \Delta X_0  & 0 \\
    0 & (a^\dag-g)(a-g) -\epsilon -g^2 - \Delta X_0^{\dag}
  \end{bmatrix},
\]
generalizing the case of the QRM. Unfortunately, solutions of the associated Riccati equation are only known explicitly for
the (initially excluded) case $\epsilon=0$, where one recovers the $J_0$ described above. It seems, the result in \cite{GD2013} is enough to show that there is a $\Z_2$-symmetry in the AQRM (for $\epsilon$ not a half-integer), and that the spectrum splits into two subsets which, however, never intersect, the hallmark of the $\Z_2$-symmetry of the QRM. Indeed, as we have explained above, the spectrum of $\HRabi{\epsilon}$ is multiplicity free for the case of non-half integer $\epsilon$ \cite{W2016JPA}. Therefore, the question arises whether the Gardas-Dajka operator $\mathbb{J}_\epsilon$ is non-trivial in the sense that it is \textit{independent} from the Hamiltonian $\HRabi{\epsilon}$.

The problem to determine whether an operator $X$ with the property $[X,H]=0$ is independent from the Hamiltonian $H$ or not has been the major obstacle to define integrability for quantum systems \cite{C2011}. At first sight, every quantum system seems to have a complete set of mutually commuting operators, namely $\hat{P}_{\lambda_i}$, the orthogonal projectors onto the eigenvectors of $H$ with eigenvalue $\lambda_i$. These operators generate obviously the commutant $\mathcal{C}(H)$ of $H$ if $H$ is self-adjoint and has a pure point spectrum. However, they cannot be the quantum analogue of the $N$ functions $L_i$ on $2N$-dimensional phase space $\mathcal{V}$ with the property $\{L_i,L_j\} =0$ in classical mechanics which render the system integrable in the sense of Liouville ($\{\cdot,\cdot\}$ denotes the Poisson bracket). For classically integrable dynamics, the $L_i$  must be independent from $L_1=H$ and from each other, which means that the 1-forms $\textrm{d}L_i$ are linearly independent at each point of the $N$-dimensional submanifold $\mathcal{M}$ of $\mathcal{V}$, characterized by $L_i=l_i$ with constants $l_i\in\R$ \cite{Ar1989}. What is the equivalent of this concept in quantum mechanics? How can one define the notion of independence for the elements of the commutant $\mathcal{C}(H)$? The operators  $\hat{P}_{\lambda_i}$ exist for every self-adjoint operator $H$ and generate $\mathcal{C}(H)$, but they are not  all independent from $H$.

The $C^\ast$ algebra $\mathcal{C}_0(H)\subset \mathcal{C}(H)$ generated by $H$ is isomorphic to $C(\Spec(H))$, the Banach algebra of continuous functions on the spectrum of $H$ \cite{R1972}.
Indeed, if $\chi_{\lambda_j}(x)$ is a continuous function with support in an interval containing only the non-degenerate eigenvalue $\lambda_j$, the projector $\hat{P}_{\lambda_j}$ is given as
$\hat{P}_{\lambda_j}=\chi_{\lambda_j}(H)$. Now we can construct ``symmetry generators'' from these projectors with arbitrary properties. For example, it is  possible to define an operator $\mathbb{J}_H$ with the properties $\mathbb{J}_H^2=I$ and $[\mathbb{J}_H,H]=0$ by dividing the spectrum of $H$ into two subsets $S_+$ and $S_-$ and define
\[
\mathbb{J}_H=\sum_{\lambda\in S_+} \chi_\lambda(H) - \sum_{\lambda\in S_-} \chi_\lambda(H).
\]
This procedure is possible only if the spectrum of $H$ is multiplicity free which is exactly the case where the construction by Gardas and Dajka works. It is not related to any special property of $H$ and clearly indicates no hidden symmetry as it applies to any multiplicity free self-adjoint operator.

On the other hand, the presence of at least one spectral degeneracy is sufficient for the existence of an operator $X\in\mathcal{C}(H)$ which is not generated by $H$ itself, $X\notin \mathcal{C}_0(H)$. But such a degeneracy should not be necessary for the existence of the non-trivial symmetry generator. The operator $X$ should still be present if the parameters of the system are continuously varied such that $H$ becomes multiplicity free, because only then it can be associated with intersections of the spectral graph. Therefore, we must consider not $\mathcal{C}_0(H)$ but a certain subalgebra which does not contain $X$ even if the spectrum is non-degenerate. Similarly, we restrict $\mathcal{C}(H)$ in a suitable way to a subalgebra $\mathcal{C}_a(H)$. Below, we shall tentatively identify $\mathcal{C}_a(H)$ with $\mathcal{C}(H)\cap\mathcal{A}$, where $\mathcal{A}$ is the algebra of analytic functions in the operators $a$ and $a^\dag$, tensored with the algebra ${\rm Mat}_2(\C)$ of $2\times2$-matrices over $\C$. Specifically, we will consider as candidates for $J_\epsilon$ the products $\mathcal{P}Q^{\epsilon}$, where $Q^\epsilon$ is a matrix-valued polynomial in $\mathcal{A}$.

In a recent work by Mangazeev, Batchelor and Bazhanov \cite{MBB2020}, for $\ell \in \Z_{\ge0}$ a general procedure is presented to obtain an operator
$J_{\frac{\ell}2}$ such that
\[
  [\HRabi{\frac{\ell}2},J_{\frac{\ell}2}] = 0,
\]
where $J_{\frac{\ell}2}$ has indeed the form given above. The authors work out examples for the cases $\ell=0,1,2$. The explicitly computed cases are such that $J_{\frac{\ell}2}^2$ is a polynomial $p_{\frac\ell2}(\HRabi{\frac{\ell}2}; g, \Delta)$ in $\HRabi{\frac{\ell}2}$ of degree $\ell$. While this method provides a general procedure for the computation, little is known on the general form of $J_{\frac{\ell}2}$ (or its square) and its properties. The authors expect that the operator are given in the form $J_{\frac{\ell}2} = \mathcal{P} Q_{\frac{\ell}2}$ for all $\ell\ge0$. The question what symmetry group should be associated with this operator was not considered in \cite{MBB2020}.

In Section \ref{sec:symmetry-aqrm} we present theoretical results regarding the operator $J_{\frac{\ell}2}$.
First, we rigorously prove its existence for any integer $\ell\geq 1$ and show that it is algebraically independent of the Hamiltonian
$\HRabi{\frac{\ell}2}$. It turns out that $J_{\frac{\ell}2}$ and $\HRabi{\frac{\ell}2}$ generate the commutant $\mathcal{C}_a(\HRabi{\frac{\ell}2})$. 
These results may be regarded as the realization of the situation described at the end of \cite{A2020}.
In addition, from these results we conclude that the expected $\mathbb{Z}_2$-symmetry, inherited from the QRM, manifests
the fact that the commuting operator $J_{\frac{\ell}2}$ is captured in the two-fold cover of the ring of polynomials $\R[\HRabi{\frac{\ell}2}]$.
In other words, the hidden symmetry is actually exhibited by the relation $J_{\frac{\ell}2}^2=p_{\frac\ell2}(\HRabi{\frac{\ell}2}; g, \Delta)$ which may be identified with the hyperelliptic curve $y^2= p_{\frac\ell2}(x,g, \Delta)$. Finally, Section \ref{sec:proofs} is devoted to
the proof of the main results .

\section{Main results}
\label{sec:symmetry-aqrm}

We state here our results. The proofs are given in Section \ref{sec:proofs}.

Our first theorem generalizes the computational observations presented in \cite{MBB2020}. 
In order to simplify the discussion, we consider the Hamiltonian $\HRabi{\epsilon}$ to be given in the equivalent form
\[
  \HRabi{\epsilon} := a^{\dag}a + \Delta \sigma_x + g (a + a^{\dag}) \sigma_z + \epsilon\sigma_z,
\]
obtained from \eqref{eq:Rabi} by means of a Cayley (unitary) transform.

Throughout this paper we denote by $\C[a,a^{\dag}]$, the Weyl algebra generated by the elements $a$ and $a^{\dag}$ and by
${\rm Mat}_2(\C[a,a^{\dag}])$ the $2\times2$ matrix algebra over $\C[a,a^{\dag}]$ . The degree of a monomial $a^k (a^\dag)^\ell \in \C[a,a^{\dag}]$
is defined to be $k+\ell$ and for a general element $\mathfrak{f} \in \C[a,a^{\dag}]$ the degree is defined to be the maximum degree of
the monomials appearing in $\mathfrak{f}$.

\begin{thm} \label{thm:existence}
  Fix $\epsilon \in  \frac12 \Z$ and set $ \ell = 2|\epsilon|$. There exists a unique (up to multiplication by constants) operator
  $Q_0^{(\epsilon)} \in {\rm Mat}_2(\C[a,a^{\dag}])$ with components of degree $\ell$ and with the properties listed below.
  \begin{enumerate}
  \item If $J_{\epsilon} := \mathcal{P} Q_0^{(\epsilon)}$, then $J_{\epsilon}$ is self-adjoint and
    \[
      [\HRabi{\epsilon},J_{\epsilon}] = 0.
    \]
  \item A normalization may be chosen so that the entries of $Q_0^{(\epsilon)}$ are polynomials in $\Q[a,a^\dag]$.
  \item The operator $Q_0^{(\epsilon)}$ is not a polynomial function of $\HRabi{\epsilon}$.
  \item There is a polynomial $p_{\epsilon} \in \C[x,g,\Delta]$ of degree $\ell$, uniquely defined up to multiplicative constant, such that
    \[
      J_{\epsilon}^2 = p_{\epsilon}(\HRabi{\epsilon};g,\Delta).
    \] 
  \end{enumerate}
\end{thm}

The theorem shows that, for $\epsilon \in \frac12 \Z$, the operator $J_{\epsilon}$ is algebraically independent of $\HRabi{\epsilon}$ while
its square $J_{\epsilon}^2$ is a polynomial on $\HRabi{\epsilon}$. We describe the relations of these facts with the expected $\Z_2$-symmetry
below after giving some further properties of the operator $J_{\epsilon}$.

For a fixed $\epsilon \in \frac12 \Z$, it is possible to find further operators $J$ that commute with $\HRabi{\epsilon}$ of the form
$J = \mathcal{P} Q$ with $Q \in {\rm Mat}_2(\C[a,a^{\dag}])$. Actually, these operators turn out to be products of the
operator $J_{\epsilon}$ and polynomials in $\HRabi{\epsilon}$.

\begin{thm} \label{thm:uni}
  For fixed $\epsilon \in \frac12 \Z$, set $\ell = 2 |\epsilon|$. Let $Q \in {\rm Mat}_2(\C[a,a^{\dag}])$ be an operator such that
  \[
    [\HRabi{\epsilon}, \mathcal{P}Q] = 0,
  \]
  then $Q = Q_0^{(\epsilon)} p(\HRabi{\epsilon})$ for some polynomial $p \in \C[x]$ and where $Q_0^{(\epsilon)}$ is the operator described in
  Theorem \ref{thm:existence}. In particular, $Q_0^{(\epsilon)}$ is the solution in  ${\rm Mat}_2(\C[a,a^{\dag}])$ with components of
  minimal degree $\ell$.
\end{thm}

Finally, for the case $\epsilon \notin \frac12\Z$ we have the following result.

\begin{prop} \label{prop:nonhi}
  If $\epsilon \notin \frac12\Z$, there is no operator $Q \in {\rm Mat}_2(\C[a,a^{\dag}])$ such that
  \[
    [\HRabi{\epsilon}, \mathcal{P}Q] = 0.
  \]
\end{prop}

In the case $\epsilon \notin \frac12\Z$, a formal solution $Q_{\epsilon}$ as a power series in $a$ and $a^\dag$ can be obtained. Clearly, such a formal solution $Q_{\epsilon}$ is not unique. Proposition \ref{prop:nonhi} is not sufficient to rule out operators $X_\epsilon=\mathcal{P}Q_\epsilon$ in $\mathcal{C}_a(\HRabi{\epsilon})$ as given above for $\epsilon \notin \frac12\Z$. However, we conjecture that all non-trivial formal solutions of $[X,\HRabi{\epsilon}]=0$ are not analytic in $a$ and $a^\dag$ if $g$ and $\Delta$ do not fulfill an additional relation such that the spectral graph as function of $\epsilon$ has (exactly) one intersection at some half-integer value of $\epsilon$.

Now, we proceed to study the expected $\Z_2$-symmetry described by $J_\epsilon$ for a half-integer $\epsilon$.  Let us first define
the algebra $\mathcal{A}'$ as generated by elements of ${\rm Mat}_2(\C[a,a^\dag])$ and $\mathcal{P}$. Furthermore,
$\mathcal{C}_a'(\HRabi{\epsilon})=\mathcal{C}(\HRabi{\epsilon})\cap\mathcal{A}'$.

In this setting, we summarize the results of Theorems \ref{thm:existence} and \ref{thm:uni} as follows.
The set of elements of $\mathcal{A}'$ commuting with $\HRabi{\epsilon}$, algebraically independent of $\HRabi{\epsilon}$ constitute the principal ideal generated by $J_\epsilon$. It then follows that
\[
  \mathcal{C}_a'(\HRabi{\epsilon}) / (J_{\epsilon},\HRabi{\epsilon}) \simeq \R,
\]
or equivalently, $\mathcal{C}_a'(\HRabi{\epsilon}) \simeq \R[J_\epsilon,\HRabi{\epsilon}]$. We notice that in addition, the equation
$J_{\epsilon}^2 = p_{\epsilon}(\HRabi{\epsilon};g,\Delta)$  gives the commutative ring $\mathcal{C}_a'(\HRabi{\epsilon})$ the
algebro-geometric structure
\begin{align}
  \label{eq:hyperelliptic}
  \mathcal{C}_a'(\HRabi{\epsilon}) \simeq \R[x,y]/(y^2 - p_{\epsilon}(x;g,\Delta)).
\end{align}
We also note that the set of solutions $Y$ in $\mathcal{A}'$ (considered as an extension of ${\rm Mat}_2(\C[a,a^\dag])$) of the equation
\[
  Y^2 - p_{\epsilon}(\HRabi{\epsilon};g,\Delta) = 0
\]
is $\{\pm J_\epsilon\}$ and that the group that permutes these solutions is exactly $\Z_2$ with the obvious action.

The $\Z_2$-symmetry may be thus considered as realized by this action and particularly the geometric structure
given by \eqref{eq:hyperelliptic}.
For $\ell \geq3$, the structure corresponds to the ring of functions of a real hyperelliptic curve (i.e. a fibered space over $\R[x]$,
the ring of functions of the projective line $\mathbb{P}^1(\R)$) and for $\ell=1,2$ it is that of a parabola or hyperbola, 
respectively.
A detailed study of the geometric properties is expected to further clarify the symmetry of the AQRM and other models with hidden
symmetry.  We further investigate the polynomials $p_{\epsilon}(x;g,\Delta)$ from the spectral degeneracy point of view in \cite{RBBW2021}. 

Finally, we note that the operator $J_{\epsilon}$ may have a non-trivial kernel $\mathcal{H}_{0}$, this corresponds to the
case that the polynomial $p_{\epsilon}(x;g,\Delta)$ vanishes for some eigenvalues of $\HRabi{\epsilon}$. By definition, the action of $J_{\epsilon}$ on
$\mathcal{H}_{0}$ is given by $J_{\epsilon}|_{\mathcal{H}_{0}}={\bm 0}$, and the commutativity of $J_{\epsilon}$ and $\HRabi{\epsilon}$ is obvious
in $\mathcal{H}_{0}$. Note that, in general, if $\mathcal{V}_{\lambda} \subset \mathcal{H}$ is the eigenspace corresponding to
$\lambda \in \Spec(\HRabi{\epsilon})$, then $J_{\epsilon} \mathcal{V}_{\lambda} \subseteq \mathcal{V}_{\lambda}$. Moreover, the case
$J_{\epsilon} \mathcal{V}_{\lambda} \neq \mathcal{V}_{\lambda}$ only occurs when $J_{\epsilon} \mathcal{V}_{\lambda} = \{0\}$, that is, when
$\mathcal{V}_{\lambda} \subset \mathcal{H}_0$. However, even in the case  $\dim \mathcal{H}_{0} \geq 1$ it holds that
$J_{\epsilon} \mathcal{V}_{\lambda} = \mathcal{V}_{\lambda}$ for almost all eigenvalues $\lambda \in \Spec(\HRabi{\epsilon})$.    
It is reasonable to expect that the $\mathcal{H}_{0} = \{ 0 \}$ for all parameter $g,\Delta>0$, however, it seems difficult to give a proof (or a counter-example) of this (see also Remark \ref{rem:curves}).

\begin{rem} \label{rem:curves}
  We would like to highlight another remarkable property of the polynomials $p_{\frac{\ell}{2}}(x;g,\Delta)$.
  As shown in Figure \ref{fig:Eigencurves}, each polynomial $p_{\frac{\ell}{2}}(x;g,\Delta)$ provides
  an excellent approximation of the first $\ell$ energy levels of $\HRabi{\frac{\ell}2}$ for the case $g/\Delta \gg 1$. 
The nature of this numerically observed relation is yet unclear, but it is another example of a property of the spectrum that holds only for
  certain parameter regimes, in this case, the deep strong coupling regime \cite{Y2017} (for an interesting discussion of
  coupling regimes for the QRM, see \cite{EVBSS2017}).

\begin{figure}[h!]
  \centering
  \subfloat[$\epsilon = \frac12$ and $\Delta=1$]{
    \includegraphics[height=3cm]{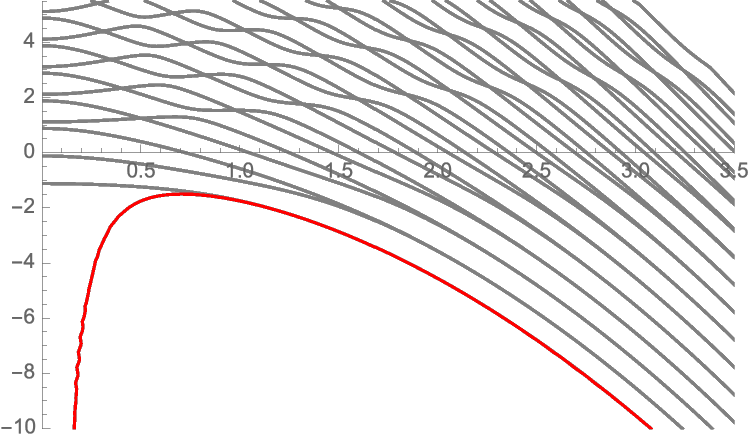}}
  ~
  \subfloat[$\epsilon = 1$ and $\Delta=1$]{
    \includegraphics[height=3cm]{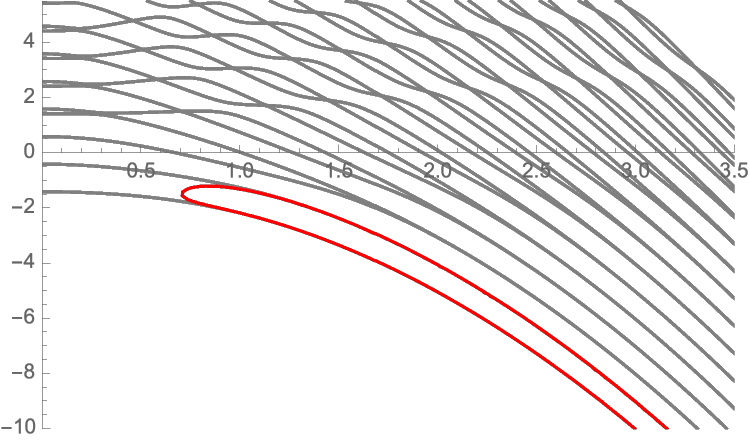}}
  \\
    \subfloat[$\epsilon = \frac32$ and $\Delta=1$]{
    \includegraphics[height=3cm]{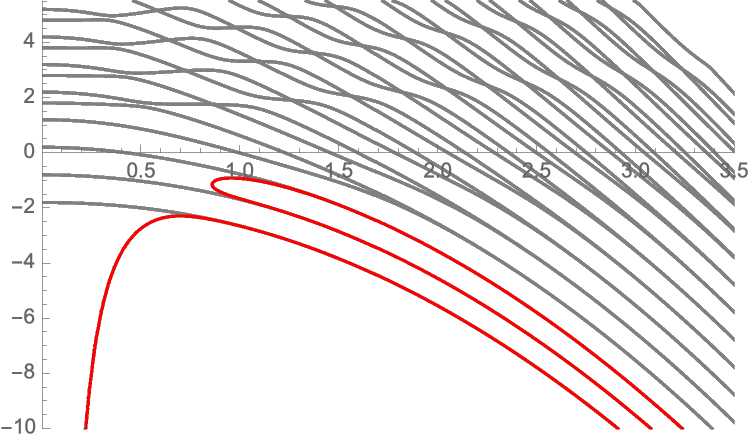}}
  ~
  \subfloat[$\epsilon = 2$ and $\Delta=1$]{
    \includegraphics[height=3cm]{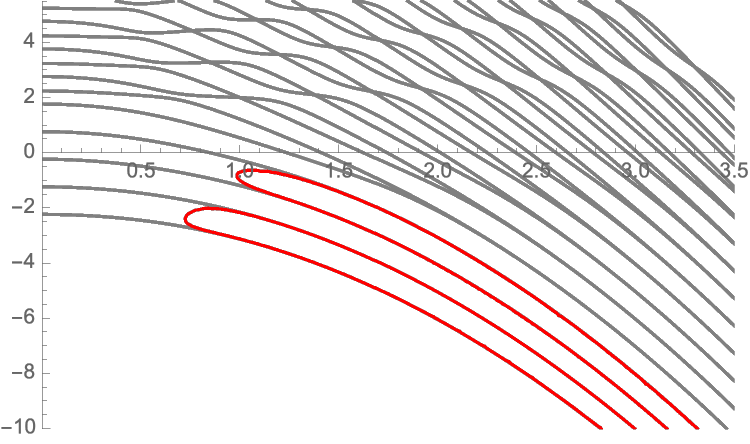}}
  \caption{Spectral curves (grey) and curves defined by $p_{\epsilon}(E,g,\Delta)=0$ (red).}
  \label{fig:Eigencurves}
\end{figure}
\end{rem}

\begin{rem}

  To further describe the hidden symmetry in a geometric way, let us consider the algebraic curve
  (hyperelliptic curves when $\ell\geq3$) described by the equation
  \begin{align} \label{eq:curve2}
    y^2 = p_{\frac{\ell}2}(x;g,\Delta).
  \end{align}
  For illustration, in Figure \ref{fig:hyperelliptic2}, we show the cases of $\ell=0,1,2,3$ for the choice of parameters $g=\Delta=1$ . 
  We note that the case $\ell=3$, after an appropriates change of variable, is given by the elliptic curve
  \begin{equation}
    y^2 =  x^3 - 5184(4g^2 - \Delta^2) x +  186624 \Delta^2.    
  \end{equation}\label{eq:elliptic}

 By the discussion of this section, the tuples of eigenvalues $(\lambda,\mu_{\lambda}) \subset \Spec(\HRabi{\frac{\ell}2})\times \Spec(J_{\frac{\ell}2})$,
 corresponding to a common eigenvector, all lie in the curve \eqref{eq:curve2}.
 The conjecture $\mathcal{H}_{0} = \{ 0 \}$ is then equivalent to the fact that no points $(\lambda,\mu_{\lambda})$ are in the
 intersection of the curve \eqref{eq:curve2} and the line $y=0$ (see Figure \ref{fig:hyperelliptic2}).

  \begin{figure}[h!]
    \centering
    \subfloat[$\ell=0$]{
      \includegraphics[height=3.5cm]{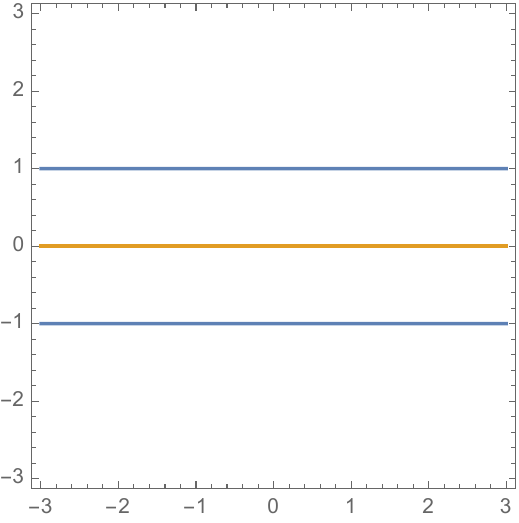}}
    ~ 
    \subfloat[$\ell=1$]{
      \includegraphics[height=3.5cm]{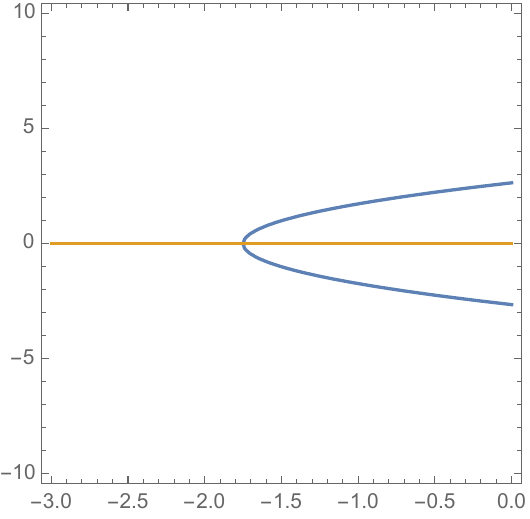}}
    ~
    \subfloat[$\ell=2$]{
      \includegraphics[height=3.5cm]{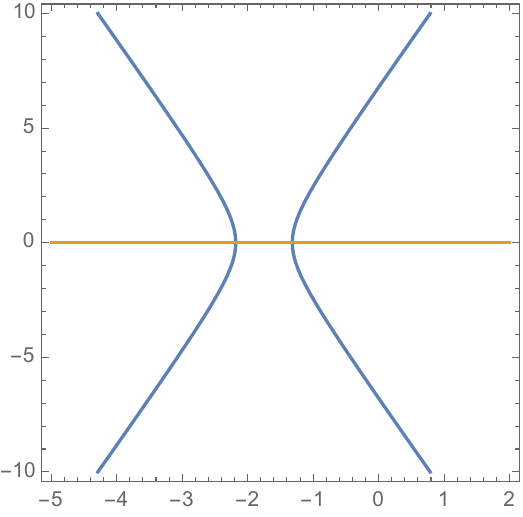}}
    ~ 
    \subfloat[$\ell=3$]{
      \includegraphics[height=3.5cm]{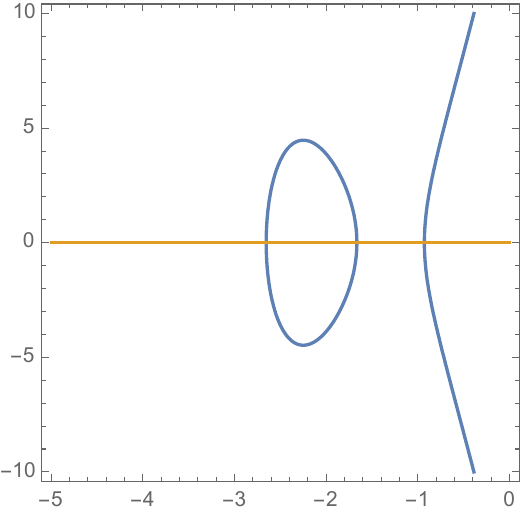}}
    \caption{Curves determined by equation \eqref{eq:curve2} for $\ell=0,1,2,3$ (blue) and the line $y=0$ (orange).}
    \label{fig:hyperelliptic2}
  \end{figure}

  Moreover, notice how in the case $\ell=0$, corresponding to the (symmetric) QRM, the geometric picture gives a clear separation of the eigenvalues into two classes, equivalent to the usual parity decomposition.
  
  For $\ell\geq1$, it is non-trivial to see the distribution of eigenvalues of $\HRabi{\frac{\ell}2}$, that is, whether each eigenvalue is
  either on the upper (positive part) or lower (negative part) part of the curves (see Figure \ref{fig:hyperelliptic2}(b-d)) even
    assuming that $\mathcal{H}_{0} = \{ 0 \}$ for all parameters.
  We leave this discussion to \cite{RBBW2021}. 
\end{rem}

\section{Proof of the main results}
\label{sec:proofs}

In this section we give the proofs of Theorems \ref{thm:existence},  \ref{thm:uni} and Proposition \ref{prop:nonhi}.

We start by noticing that 
\begin{equation*}
  \HRabi{\epsilon} \mathcal{P} = \mathcal{P}  \widetilde{\HRabi{\epsilon}},
\end{equation*}
with
\[
  \widetilde{\HRabi{\epsilon}} :=
  \begin{bmatrix}
    a^\dag a - g(a + a^\dag) + \epsilon & \Delta \\
    \Delta & a^\dag a + g (a+ g^\dag) - \epsilon
  \end{bmatrix}.
\]
Therefore, if $J = \mathcal{P} Q$ then $[\HRabi{\epsilon},J]=0$ is equivalent to
\begin{equation}
  \label{eq:algcomm}
  \widetilde{\HRabi{\epsilon}} Q-  Q \HRabi{\epsilon}  = 0.
\end{equation}

The first step is the describe the solutions of \eqref{eq:algcomm} for matrices with entries in $\C[a,a^{\dag}]$.
To do this, we assume that the operator $Q$ is given by
\[
  Q =
  \begin{bmatrix}
    \alpha(a,a^\dag) & \beta(a,a^\dag) \\
    \gamma(a,a^\dag) & \delta(a,a^\dag)
  \end{bmatrix},
\]
where $\alpha(x,y) \in \C[[x,y]]$ (power series on variables $x,y$) is given by
\[
  \alpha(x,y) = \sum_{n=0}^{\infty} \sum^{\infty}_{m=0} \alpha_{n,m} x^n y^m,
\]
with $\alpha_{n,m} \in \C$. Similar definitions are given for the other components. To simplify the discussion
we assume that polynomials in $a$ and $a^{\dag}$ are reduced to this form during the computations.
We note here that the adjoint of the formal power series $\alpha(a,a^{\dag})$ is given by
\begin{equation*} 
   \alpha(a,a^{\dag})^{\dag} = \sum_{n=0}^{\infty} \sum^{\infty}_{m=0} \overline{\alpha_{m,n}} a^n (a^{\dag})^m.
\end{equation*}

\begin{dfn}
  We say that a solution $Q$ of \eqref{eq:algcomm} is {\em polynomial} if $Q \in {\rm Mat}_2(\C[a,a^{\dag}])$, that is, if
  the power series $\alpha(a,a^\dag), \beta(a,a^\dag),\gamma(a,a^\dag),\delta(a,a^\dag)$ describing the components of $Q$ are actually polynomials. 
  Equivalently, for some $M \in \Z_{\geq0}$, we have  $\alpha_{n,m}= \beta_{n,m} = \gamma_{n,m}=\delta_{n,m}= 0$ for $n+m>M$.
\end{dfn}

\begin{rem}
  An alternative approach is to look for solutions of the form $J= Q \mathcal{P}$. The approach is completely
equivalent and the solutions obtained by this method are just the adjoint of the ones obtained by \eqref{eq:algcomm}.
\end{rem}

\begin{ex}
  Let us write some solutions for small values of half-integer $\epsilon$. For $\epsilon = 0$, a polynomial solution of
  degree $0$ given by
  \[
    Q_0 =
    \begin{bmatrix}
      0 & 1 \\
      1 & 0
    \end{bmatrix} = \sigma_x,
  \]
  and a solution of degree $2$ is given by
  \[
    Q_{0}' =
    \begin{bmatrix}
     \Delta  & a a^{\dag} - g (a+a^{\dag}) \\
     a a^{\dag} + g (a + a^{\dag}) & \Delta
    \end{bmatrix}.
  \]
  Notice that $Q_0'= \sigma_x( \HRabi{0}+ {\bm I}).$
  
  For $\epsilon = \frac{1}{2}$, a polynomial solution of degree $1$ is given by
  \[
    Q_{\frac12} =
    \begin{bmatrix}
      \Delta & 2 g (g-a) \\
      2g (g+a^{\dag}) & \Delta
    \end{bmatrix}.
  \]
  For $\epsilon = 1$, a polynomial solution of degree $2$ is given by
  \[
    Q_1 =
    \begin{bmatrix}
-2 a g \Delta + 2 a^{\dag} g \Delta + 4 g^{2} \Delta + \Delta & 4 a^{2} g^{2} - 8 a g^{3} + 4 g^{4} + \Delta^{2} \\
4 (a^{\dag})^{2} g^{2} + 8 a^{\dag} g^{3} + 4 g^{4} + \Delta^{2} & -2 a g \Delta + 2 a^{\dag} g \Delta + 4 g^{2} \Delta - \Delta
    \end{bmatrix}.
  \]
   For $\epsilon = \frac{3}2$, a polynomial solution of degree $3$ is given by $Q_{\frac32} = \begin{bsmallmatrix}q_{1,1} & q_{1,2}\\q_{2,1}&q_{2,2}\end{bsmallmatrix}$,
  with
  \begin{align*}
    q_{1,1} &= -2\Delta - 8g^2\Delta - 12g^4\Delta - \Delta^3 - 4g\Delta a^{\dag} - 12 g^3\Delta a^\dag - 4g^2\Delta (a^\dag)^2 + 4 g \Delta a+ 12 g^3 \Delta a + 4 g^2 \Delta a a^\dag
              - 4 g^2 \Delta a^2, \\
    q_{1,2} &=  -8 g^6 - 6 g^2 \Delta^2 - 2 g \Delta^2 a^\dag + 24 g^5 a + 4 g \Delta^2 a - 24 g^4 a^2 + 8 g^3 a^3,\\
    q_{2,1} &= -8 g^6 - 6 g^2 \Delta^2 + 2 g \Delta^2 a - 24 g^5 a^\dag - 4 g \Delta^2 a^\dag - 24 g^4 (a^\dag)^2 - 8 g^3 (a^\dag)^3,\\
    q_{2,2} &= -2 \Delta + 4 g^2 \Delta  - 12 g^4 \Delta - \Delta^3 + 4 g \Delta a^{\dag} - 12 g^3 \Delta a^\dag - 4 g^2 \Delta (a^\dag)^2 - 4 g \Delta a + 12 g^3\Delta a + 4 g^2 \Delta a a^\dag
              - 4 g^2\Delta a^2.
  \end{align*}
\end{ex}
  
\begin{ex} \label{ex:J2}
  
  Let $J_{i} = \mathcal{P}Q_i$, for $i=0,\frac12,1,\frac32$ as in the examples above. We have
  \begin{align*}
    J_{0}^2 &= {\bm I},  \\ 
    J_{\frac12}^2 &=  4g^2 \HRabi{\frac12} + (4g^4+2g^2+\Delta^2) {\bm I}, \\
    J_{1}^2 &=  16 g^4 (\HRabi{1})^2+ 8g^2(4 g^4+2 g^2+\Delta^2)\HRabi{1} + (16 g^{8} + 16 g^{6} + 8 g^{4} \Delta^{2} + 4 g^{2} \Delta^{2} + \Delta^{4} + \Delta^{2}) {\bm I}, \\
    J_{\frac32}^2 &= 64 g^6(\HRabi{\frac32})^3 + 48g^4(4g^4+2g^2+\Delta^2)(\HRabi{\frac32})^2+ 4g ^2(48 g^8+48 g^6+24 g^4 \Delta^2-16 g^4+12 g^2 \Delta^2+3 \Delta^4+4 \Delta^2) \HRabi{\frac32}\\
            &+(64g^{12}+96g^{10}+48g^8\Delta^2-16g^8+48g^6\Delta^2-24 g^6 +12g^4\Delta^4 +12g^4\Delta^2+6g^2\Delta^4+\Delta^6+8g^2\Delta^2+4\Delta^4+4\Delta^2){\bm I}.
  \end{align*}

\end{ex}

With these preparations, let us proceed to prove Theorems \ref{thm:existence} and \ref{thm:uni}. We do this by
proving the individual statements separately. 

\begin{prop} \label{prop:existencePoly}
  Let $\epsilon = \frac12 \Z$ and set $\ell = 2|\epsilon|$. Then there exists a polynomial solution $Q_0^{(\epsilon)}$ of degree $\ell$ of \eqref{eq:algcomm}.
  Moreover, up to multiplication by constants, $Q_0^{(\epsilon)}$ is the unique solution of degree $\ell$. In addition, there are no polynomial
  solutions of \eqref{eq:algcomm} of degree smaller than $\ell$.
\end{prop}

\begin{proof}
 The condition \eqref{eq:algcomm} is equivalent to
  \[
    2\epsilon \begin{bmatrix}
      0 & \beta \\
      -\gamma & 0
    \end{bmatrix}
    +
    \begin{bmatrix}
      [a^\dag a,\alpha] & [a^\dag a,\beta] \\
      [a^\dag a,\gamma] & [a^\dag a,\delta]
    \end{bmatrix}
    +g
    \begin{bmatrix}
      -(a+a^\dag )\alpha - \alpha(a+a^\dag)  &  - [a+a^\dag,\beta]  \\
      [a+a^\dag,\gamma]  & (a+a^\dag)\delta +\delta (a+a^\dag)
   \end{bmatrix}
       +
    \Delta
    \begin{bmatrix}
      \gamma - \beta & \delta-\alpha \\
      \alpha-\delta & \beta-\gamma
    \end{bmatrix}=0.
  \]

  The equality defines a set of simultaneous recurrence relation in terms of the coefficients of the polynomials
  $\alpha,\beta,\gamma,\delta$. For $n,m \geq 0$ the general term of then recurrence relations are given as follows
\begin{equation}
  \label{eq:rec11}
  (m-n) \alpha_{n,m} - 2g \alpha_{n-1,m} -2g \alpha_{n,m-1} + (m+1) g \alpha_{n,m+1} + (n+1)g \alpha_{n+1,m} + \Delta(\gamma_{n,m}-\beta_{n,m}) =0,
\end{equation}
\begin{equation}
  \label{eq:rec12}
  (2\epsilon - (n-m)) \beta_{n,m} - (m+1) g \beta_{n,m+1} + (n+1)g \beta_{n+1,m} + \Delta (\delta_{n,m} - \alpha_{n,m}) = 0,
\end{equation}
\begin{equation}
  \label{eq:rec21}
    (-2\epsilon - (n-m)) \gamma_{n,m} + (m+1) g \gamma_{n,m+1} - (n+1)g \gamma_{n+1,m} + \Delta ( \alpha_{n,m}-\delta_{n,m} ) = 0,
\end{equation}
\begin{equation}
  \label{eq:rec22}
    (m-n) \delta_{n,m} + 2g \delta_{n-1,m} +2g \delta_{n,m-1} - (m+1) g \delta_{n,m+1} - (n+1)g \delta_{n+1,m} + \Delta(\beta_{n,m}-\gamma_{n,m}) =0,
\end{equation}
with initial conditions $\alpha_{n,m}=0$ for $n,m<0$ and similar for the other coefficients.  Consider a fixed integer $N\geq 0$.
The condition that $Q$ is a polynomial solution of degree $N$ imposes the additional condition $\alpha_{n,m}=0$ for $n+m>N$.

First, by considering \eqref{eq:rec11} and \eqref{eq:rec22} for $n+m= N+1$ we see that
\[
  \alpha_{n-1,m} +\alpha_{n,m-1}=0, \qquad  \delta_{n-1,m} +\delta_{n,m-1}=0
\]
for $n= 0,1,\cdots,N$. In particular, $\alpha_{0,N}= \delta_{0,N} = 0$ and it follows that $\alpha_{n,m}=0$ (resp. $\delta_{n,m}=0$ ) for $n+m=N$.

Next, we consider \eqref{eq:rec12} and \eqref{eq:rec21} for  $n+m=N$. The recurrence relations for this case reduce to
\[
  (2\epsilon - (n-m)) \beta_{n,m} = 0, \qquad (-2\epsilon - (n-m)) \gamma_{n,m}   = 0.
\]
Note that if $\epsilon \neq \pm \frac{n-m}{2}$ for some $n,m\geq0$  with $n+m=N$, $\beta_{n,m} = \gamma_{n,m} = 0$ for all $n,m\geq0$  with $n+m=N$, and
we are reduced to the case of polynomial solutions of degree $N-1$. In particular, if $\epsilon \notin \frac12 \Z$, then the only polynomial
solution (of any degree) is the trivial solution $Q=0$. Therefore, for a polynomial solution of degree $N$ we must have $\epsilon  = \frac{N-2i}2$ with $i=0,\ldots,N$.

Suppose that $\epsilon  = \frac{N-2i}2$, then for $i=0,\cdots,N$ the values $\beta_{N-i,i}$ and $\gamma_{i,N-i}$ are arbitrary, and $\beta_{n,m}=\gamma_{m,n}=0$
for $n\neq N-i$ and $m\neq i$ with $n+m=N$. The coefficients for the cases $n,m\geq 0$ with $n+m<N$ can then be computed then according to the
recurrence relations above. We note in particular, that for the case $n+m = N$, the recurrence relations \eqref{eq:rec11} and \eqref{eq:rec22} give
\[
  \alpha_{n-1,m} + \alpha_{n,m-1} =  \frac{\Delta}{2g}(\gamma_{n,m}-\beta_{n,m}),
\]
and a similar one for the case of $\delta_{n,m}$. By a similar argument than the one used for the case of $n+m=N+1$, we see that
all the values of $\alpha_{n,m}$ for $n,m\geq$ with $n+m=N-1$ can be computed and furthermore, we see that
\[
  \gamma_{i,N-i} = (-1)^{N-2i} \beta_{N-i,i}
\]
and therefore, up to a constant factor, the leading coefficients of the polynomials are determined uniquely.

Note that for $n+m<\ell$, the conditions of \eqref{eq:rec12} (resp. \eqref{eq:rec21}) are
\[
  \beta_{n,m} = \frac{1}{2\epsilon - (n-m)}\left((m+1) g \beta_{n,m+1} - (n+1)g \beta_{n+1,m} - \Delta (\delta_{n,m} - \alpha_{n,m}) \right)
\]
for the case $\epsilon \neq \frac{n-m}{2}$. In the case $\epsilon = \frac{n-m}{2}$ the coefficient $\beta_{n,m}$ is free and can take any arbitrary value.
The same situation holds for $\gamma_{n,m}$. Continuing this process in the same manner, using the recurrence relations above, the remaining coefficients can be computed to obtain a polynomial solution $Q$. This proves the existence of polynomial solutions for any
$\epsilon \in \frac12 \Z$. Moreover, the polynomial solutions are of degree $\ell+2k$ for $k\geq 0$, and in the case of degree
$\ell$ there are no degrees of freedom on the coefficients with exception of the leading term of the polynomial. This argument completes the proof.

\end{proof}

By the proof of Proposition \ref{prop:existencePoly}, for $\epsilon \notin \frac12 \Z$, there cannot be a polynomial solution of
\eqref{eq:algcomm} of any finite degree, thus Proposition \ref{prop:nonhi} is proved.
Also, notice from the recurrence relations on the proof of Proposition \eqref{prop:existencePoly} that by taking
the leading coefficient to be a integer, all the coefficients of the polynomials are in $\Q[g,\Delta]$.

\begin{prop}
  For $\epsilon \in \frac12 \Z$ with $\ell = 2|\epsilon|$, the operator $Q_0^{(\epsilon)}$ of degree $\ell$ is not a polynomial in $\HRabi{\epsilon}$.
\end{prop}

\begin{proof}
  Suppose that $Q_0^{(\epsilon)}= p(\HRabi{\epsilon})$ for $p \in \C[x]$. The relation \eqref{eq:algcomm} reduces to
  \[
    (\widetilde{\HRabi{\epsilon}}- \HRabi{\epsilon})p(\HRabi{\epsilon})= 0.
  \]
  Note that $\widetilde{\HRabi{\epsilon}}- \HRabi{\epsilon} = - 2g (a+a^{\dag}) \sigma_z$, thus, since $g \neq 0$ then
  $2g(a+a^{\dag})$ must be an annihilator of each of the components of $p(\HRabi{\epsilon})$, which is impossible
  since the monomials of $2g(a+a^{\dag})$ have the same degree and coefficients and $p(x)$ is a polynomial. Therefore,
  we must have $p \equiv 0$, contradicting the fact that $Q_0^{(\epsilon)}$ is a solution of degree $\ell$.
\end{proof}

Next, we set $J_{\epsilon} := \mathcal{P} Q_0^{(\epsilon)}$. By the discussion above, $J_{\epsilon}$ satisfies
\[
  [\HRabi{\epsilon},J_{\epsilon}] = 0.
\]

\begin{prop}
  For $\epsilon \in \frac12 \Z$, the operator $J_{\epsilon}$ is self-adjoint.
\end{prop}

\begin{proof}
  The statement is equivalent to the following polynomial identities
  \begin{align*}
    \alpha(a,a^{\dag}) &= \alpha(-a,-a^{\dag})^{\dag} \\
    \delta(a,a^{\dag}) &= \delta(-a,-a^{\dag})^{\dag} \\
    \beta(a,a^{\dag}) &= \gamma(-a,-a^{\dag}),
  \end{align*}
  which, in turn, results in the equalities of coefficients
  \begin{align*}
    \alpha_{n,m} &= (-1)^{n+m} \alpha_{m,n}\\
    \delta_{n,m} &= (-1)^{n+m} \delta_{m,n}\\
    \beta_{n,m} &= (-1)^{n+m} \gamma_{m,n}.
  \end{align*}
  In the above equations we have assumed the coefficients to be real numbers and thus we
  have omitted the complex conjugate. This does not represent a loss in generality, since in the general case, we have,
  for example
  \[
    \alpha_{n,m} = (-1)^{n+m} \overline{\alpha_{m,n}},
  \]
  however, since the coefficients $\alpha_{n,m}$ and $\alpha_{m,n}$ are given by a real number multiplied by a common constant, the equation
  forces the common constant to be a real number.

  Next, we set $\ell = 2|\epsilon|$ and note that by the proof of Proposition \ref{prop:existencePoly} the equalities holds for $n,m \geq 0$
  with $n+m \ge \ell$ (which are only non-vanishing for the case $n+m=\ell$. Now, let us take $N \geq 0$ and suppose the result holds
  for $n,m \geq 0$ with $n+m \geq N$,
  then, by the recurrence relation \eqref{eq:rec12} that for $n,m\geq0$ with $n+m=N-1$ we have
  \begin{align*}
    (2\epsilon - (m-n)) \beta_{m,n} &= (n+1)g \beta_{m,n+1} - (m+1)g \beta_{m+1,n} - \Delta (\delta_{m,n} - \alpha_{m,n}) \\
                         &= (-1)^{N-1}\left( - (n+1)g \gamma_{n+1,m} + (m+1)g \gamma_{n,m+1}  + \Delta (\alpha_{n,m} - \delta_{n,m}) \right)\\
                         &= (-1)^{N-1} (2\epsilon+(n-m))\gamma_{n,m},
  \end{align*}
  where the last equality holds by \eqref{eq:rec21}. Therefore $\beta_{m,n} =  (-1)^{N-1}  \gamma_{n,m}$, as desired.

  Let us now consider the case of $\alpha_{n,m}$. We omit the proof for $\delta_{n,m}$ since it is completely analogous. As usual, first, we
  consider the extremal cases $n=0$ or $m=0$. By the recurrence relations \eqref{eq:rec11} and \eqref{eq:rec22}, we have
  \begin{align*}
    2g \alpha_{0,N-1}   &= N \alpha_{0,N}  + (N+1) g \alpha_{0,N+1} + g \alpha_{1,N} + \Delta(\gamma_{0,N} - \beta_{0,N}) \\
                   &= (-1)^{N-1} \left( - N \alpha_{N,0} + (N+1)g \alpha_{N+1,0} + g \alpha_{N,1} + \Delta(\gamma_{N,0} - \beta_{N,0}) \right) \\
                   &= (-1)^{N-1} (2g \alpha_{N-1,0}),
  \end{align*}
  giving the desired equality. The remaining cases of $n+m = N-1$ are dealt in the same way. 
\end{proof}

By the considerations given in the proof of the proposition above, from this point we assume that the constant in $Q_{0}^{(\epsilon)}$
is chosen so that it has real coefficients.

Next, we deal with the issue of the uniqueness of the solution $ Q_0^{(\epsilon)}$ given above. It is clear that the addition of two
solutions of \eqref{eq:algcomm} is another solution and that the same is true for multiplication by real constants.

Denote by $V_n^{(\epsilon)}$ the real vector space of polynomial solutions of \eqref{eq:algcomm}  of degree smaller or equal to $n$. Clearly, for
$\epsilon \in \frac12 \Z$ and $\ell= 2|\epsilon|$, by the proposition and corollary we see that 
\begin{align*}
  \dim V_{\ell}^{(\epsilon)} &= 1, \\
  \dim V_{n}^{(\epsilon)} &= 0
\end{align*}
for $n < \ell$.

\begin{prop}
  Let  $\epsilon \in \frac12 \Z$ and set $\ell = 2|\epsilon|$. We have
  \[
    \dim V_{n+1}^{(\epsilon)} = \dim V_{n}^{(\epsilon)},
  \]
  for $n = \ell+2 k$ for some $k\geq 0$, and
  \[
    \dim V_{m+1}^{(\epsilon)} = \dim V_{m}^{(\epsilon)}+1
  \]
  for $m= \ell+2 k +1$ for some $k\geq 0$.
\end{prop}

\begin{proof}
  Suppose that $n= \ell+2k$ for some $k$ and a polynomial solution of \eqref{eq:algcomm} of degree $n+1$. By the proof of Proposition \ref{prop:existencePoly} and in particular, recurrence relations \eqref{eq:rec12} and \eqref{eq:rec21}, we see that all the coefficients of degree $n+1$ vanish and we are reduced to a polynomial solution of degree $n$. Therefore, $\dim V_{n+1}^{(\epsilon)} = \dim V_{n}^{(\epsilon)}$.

  On the other hand, let us consider the case $m= \ell+2 k +1$, in this case there are nonzero polynomial solutions of degree $m+1$.
  Let us denote a fixed solution of degree $m+1$ by $A$ normalized to have leading coefficients $\pm 1$. Clearly, if $S$ is a basis of
  $\dim V_{m}^{(\epsilon)}$ the set $S \cup \{A\}$ is a linear independent set in $\dim V_{m+1}^{(\epsilon)}$.

  Next, let $B \in  \dim V_{m+1}^{(\epsilon)}$, that is, $B$ is an arbitrary solution of \eqref{eq:algcomm} of degree at most $m+1$. First, if
  the degree of $B$ is $m$ or smaller, then it is generated by the elements of $S$. On the other hand, if the degree $B$ is  exactly $m+1$,
  then by the recurrence relations \eqref{eq:rec12} and \eqref{eq:rec21}, the coefficients of degre $m+1$ and $m$ are determined up
  to a constant factor $0 \neq \alpha \in \C$. Therefore, $B-\alpha A = C$ is a polynomial solution of degree $m-1$ and $B = \alpha A + C$. This proves
  that $S \cup \{A\}$ is a basis of $\dim V_{m+1}^{(\epsilon)}$ and the result follows. 
\end{proof}

The reason we introduce the vector space $V_{n}^{(\epsilon)}$ is that for higher degrees ($n \geq 2\epsilon$) the corresponding uniqueness statement
(up to constant) of Proposition \ref{prop:existencePoly} does not hold since arbitrary linear combinations of smaller degree solutions can be added to any given solution.

\begin{prop}
  For $n \ge 0$ and $\epsilon \in \frac12 \Z$ with $\ell = 2|\epsilon|$, the operator $Q_0^{(\epsilon)} (\HRabi{\epsilon})^{n}$ is a polynomial solution of degree
  $\ell+2n$ of \eqref{eq:algcomm}.
\end{prop}

\begin{proof}
  The result immediately follows from
  \[
    \widetilde{\HRabi{\epsilon}} Q_0^{(\epsilon)} (\HRabi{\epsilon})^{n} = Q_0^{(\epsilon)}\HRabi{\epsilon} (\HRabi{\epsilon})^{n} = Q_0^{(\epsilon)} (\HRabi{\epsilon})^{n}\HRabi{\epsilon}.
  \]
\end{proof}

\begin{cor} \label{cor:structurePoly}
  Let  $\epsilon \in \frac12 \Z$, $N\geq 0$ and set $\ell = 2|\epsilon|$ . The set
  \[
    \left\{ Q_0^{(\epsilon)},Q_0^{(\epsilon)}\HRabi{\epsilon}, Q_0^{(\epsilon)}(\HRabi{\epsilon})^{2}, \cdots, Q_0^{(\epsilon)}(\HRabi{\epsilon})^{N} \right\},
  \]
  is a basis of $V^{(\epsilon)}_{\ell+2N}$. In particular, we have $\dim V^{(\epsilon)}_{\ell+2N} = N+1$. \qed
\end{cor}

Therefore, for $\epsilon \in \frac12 \Z$ any polynomial solution $Q$ of \eqref{eq:algcomm} is of the form $Q= Q_{0}^{(\epsilon)}p(\HRabi{\epsilon})$ for
certain polynomial $p \in \C[x]$. Next, by using Corollary \ref{cor:structurePoly}, we show that the square of the operator
$J_{0}^{(\epsilon)}$ is a polynomial in $\HRabi{\epsilon}$.

We need the following lemma.

\begin{lem} \label{lem:annQ}
  For $\epsilon \in \frac12 \Z$, the matrix $Q_{0}^{(\epsilon)}$ has no non-trivial polynomial right annihilators. That is, there is no $A \in {\rm Mat}_2(\C[a,a^{\dag}])$ such that $Q_{0}^{(\epsilon)} A = 0$ except for the zero matrix. 
\end{lem}

\begin{proof}
  Suppose $A \in {\rm Mat}_2(\C[a,a^{\dag}])$ is a right annihilator of $Q_{0}^{(\epsilon)}$ given entrywise by polynomials $(A)_{i,j}= A_{i,j}(a,a^{\dag})$
  of fixed degree $N$.
  Let us write $\ell = 2|\epsilon|$, then by the proof of Proposition \ref{prop:existencePoly} after a normalization, we can write
  \[
     Q_{0}^{(\epsilon)} =
     \begin{bmatrix}
       0 & a^{\ell} \\
       (-1)^\ell (a^\dag)^{\ell} & 0
     \end{bmatrix}
     + \overline{Q_0},
   \]
   where $\overline{Q_0}$ is a polynomial matrix of degree strictly less than $\ell = 2 |\epsilon|$.
   We then have
   \[
     0 = Q_{0}^{(\epsilon)} A =
     \begin{bmatrix}
       a^{\ell} A_{2,1} & a^{\ell} A_{2,2} \\
       (-1)^{\ell} (a^{\dag})^{\ell} A_{1 ,1} & (-1)^{\ell} (a^{\dag})^{\ell} A_{1 ,2}
     \end{bmatrix}
     + \overline{Q_0} A.
   \]
   Let us write
   \[
     A_{2,1}(a,a^{\dag}) = \sum_{i=0}^{N} c^{(N)}_i a^i (a^{\dag})^{N-i} + \overline{A_{2,1}}(a,a^\dag),
   \]
   where $\deg \overline{A_{2,1}} < N$.
   Let us now consider
   \[
     a^{\ell} A_{2,1} = \sum_{i=0}^{N} c^{(N)}_i a^{\ell+i} (a^{\dag})^{N-i}  + a^{\ell} \overline{A_{2,1}},
   \]
   since $\deg a^{\ell} \overline{A_{2,1}} < N+\ell$ and $\deg \overline{Q_0} A < N+\ell$, we have
   \[
     \sum_{i=0}^{N} c^{(N)}_i a^{\ell+i} (a^{\dag})^{N-i} = 0,
   \]
   which implies $c^{(N)}_i = 0$, and $\deg A_{2,1} \leq N-1$. By repeated application of this procedure
   we conclude that $A_{2,1} =0$. The cases of $A_{1,1}$,$A_{1,2}$ and $A_{2,2}$ are analogous, and we
   conclude that $A=0$, as desired.
\end{proof}

\begin{prop} \label{prop: non-zero}
  For $\epsilon\in \frac12 \Z$, we have $J_{\epsilon}^2 = p_{\epsilon}(\HRabi{\epsilon};g,\Delta)$, where $p_{\epsilon} \in \C[x,g,\Delta]$ of degree $\ell =2|\epsilon|$ on the variable $x$.
\end{prop}

\begin{proof}
  First, it is clear that
  \( J_{\epsilon}^2 = \left(Q_{0}^{(\epsilon)}\right)^\dag Q_{0}^{(\epsilon)} \in {\rm Mat}_2(\C[a,a^{\dag}])\).
  Then, let us consider
  the operator
  \[
    J_{\epsilon}^3 = \mathcal{P} Q_{0}^{(\epsilon)} J_{\epsilon}^2 = \mathcal{P} Q_{0}^{(\epsilon)} \left(Q_{0}^{(\epsilon)}\right)^\dag Q_{0}^{(\epsilon)}.
  \]
  Since $J_{\epsilon}^3$ commutes with $\HRabi{\epsilon}$, we see that $Q_{0}^{(\epsilon)} J_{\epsilon}^2$ is a polynomial solution of \eqref{eq:algcomm}
  and by Corollary \ref{cor:structurePoly} we have
  \[
    Q_{0}^{(\epsilon)} J_{\epsilon}^2 = Q_{0}^{(\epsilon)} p(\HRabi{\epsilon}),
  \]
  it follows that
  \[
    Q_{0}^{(\epsilon)}\left( J_{\epsilon}^2 - p(\HRabi{\epsilon})\right) = 0,
  \]
  and we have $J_{\epsilon}^2 = p(\HRabi{\epsilon}) $ by Lemma \ref{lem:annQ}. Since the degree of $J_{\epsilon}^2$ as a polynomial solution is $2\ell$ the
  degree of $p_{\epsilon}$ must be $\ell$.
\end{proof}

\begin{rem}

  In this paper when we refer to the polynomial $p_{\pm \frac{\ell}{2}}(x;g,\Delta)$ we consider the normalization such that h
  the leading term is equal to $(2g)^{2\ell}$ (see Example \ref{ex:PolyP}).
\end{rem}

\begin{ex} \label{ex:PolyP}
  The first few polynomials $p_{\epsilon}(x;g,\Delta)$ are given by
  \begin{align*}
    p_{0}(x;g,\Delta) &= 1,  \\ 
    p_{\frac12}(x;g,\Delta) &=  (2g)^2x + (4g^4+2g^2+\Delta^2),\\
    p_{1}(x;g,\Delta) &=  p_{\frac12}(x;g,\Delta)^2 - 4 g^4 + \Delta^2,\\
    p_{\frac32}(x;g,\Delta) &=  p_{\frac12}(x;g,\Delta)^3 + (- 28 g^4 + 4 \Delta^2)p_{\frac12}(x;g,\Delta) + 4 (6 g^6 + 12 g^8 + \Delta^2 + 3g^4 \Delta^2).
  \end{align*}
\end{ex}

\medskip

Finally, we complement the discussion by showing that the constant term of the polynomial $p_\epsilon(x;g,\Delta)$ is not the trivial
polynomial in $g$ and $\Delta$.

\begin{thm} 
  Let $\epsilon \in \frac12 \Z$ and $p_\epsilon(x;g,\Delta)$ the polynomial of Proposition \ref{prop: non-zero}. Then, for any $\alpha \in \R$, $p_{\epsilon}(\alpha;g,\Delta)$ is not identically $0$.
\end{thm}

\begin{proof}

  Let $\ell = 2|\epsilon|$. It is sufficient to consider the case $\alpha=0$, since if an operator commutes with $\HRabi{\epsilon}$, then it also commutes with $\HRabi{\epsilon} - \alpha$ (and define the same $J_\epsilon$). In additions, since $J_\epsilon^2$ is a polynomial in $\HRabi{\epsilon}$, it is enough to show that
  $\HRabi{\epsilon}$ does not divide $J_\epsilon^2$ from the right.
  The case $\ell=0$ is obvious and the case $\ell=\frac12$ can be verified directly so we assume that $\ell>1$.

  Suppose that  $ J_{\epsilon}^2  = \bm{M} \HRabi{\epsilon} $ with $\bm{M} \in {\rm Mat}_2(\C[a,a^{\dag}]) $.
  We assume that $Q_0^{(\epsilon)}$ is given in the normalization where the coefficients of $a^{\ell}$ and $(a^{\dag})^{\ell}$ are $\pm 1$.

  Let us denote the components of $Q_0^{(\epsilon)}$  by $\alpha,\beta,\gamma,\delta$ as before. By the recurrence relations of Proposition \ref{prop:existencePoly}
  we see that
  \begin{align*}
    \alpha(a,a^{\dag}) &= (-1)^{\ell} \frac{\Delta}{2 g} \left( (a^{\dag})^{\ell-1} - a (a^\dag)^{\ell-1} +  \cdots + (-1)^\ell a^{\ell-1}\right) + \alpha_1(a,a^{\dag}), \\
    \gamma(a,a^\dag) &= (-1)^{\ell} \left( (a^{\dag})^\ell + \ell g (a^{\dag})^{\ell-1} \right) + \gamma_1(a,a^{\dag}),
  \end{align*}
  where $\alpha_1(x,y)$ and $\gamma_1(x,y)$ are polynomials of degree $\ell-2$ or smaller.

  Next, we note that
  \[
    J_{\epsilon}^2 = (Q_{0}^{(\epsilon)})^\dag Q_{0}^{(\epsilon)} =
    \begin{bmatrix}
      \alpha(a,a^{\dag})^{\dag} & \gamma(a,a^{\dag})^{\dag} \\
      \beta(a,a^{\dag})^{\dag} & \delta(a,a^{\dag})^{\dag}
    \end{bmatrix}
    \begin{bmatrix}
      \alpha(a,a^{\dag}) & \beta(a,a^{\dag}) \\
      \gamma(a,a^{\dag}) & \delta(a,a^{\dag})
    \end{bmatrix},
  \]
  and in particular,
  \[
    (J_{\epsilon}^2)_{1,1} = \alpha(a,a^{\dag})^{\dag} \alpha(a,a^{\dag}) + \gamma(a,a^{\dag})^{\dag} \gamma(a,a^{\dag}).
  \]

  The product $\alpha(a,a^{\dag})^{\dag} \alpha(a,a^{\dag})$ is equal to
  \begin{align*}
    &\left((-1)^{\ell} \frac{\Delta}{2 g}  \sum_{i=0}^{\ell-1}(-1)^i a^{\ell-1-i}(a^{\dag})^{i} + \alpha_1(a,a^{\dag})^{\dag}\right) \left((-1)^{\ell} \frac{\Delta}{2 g}  \sum_{i=0}^{\ell-1}(-1)^i a^{i}(a^{\dag})^{\ell-1-i} + \alpha_1(a,a^{\dag})\right) \\
    &\quad = \frac{\Delta^2}{4 g^2} \left( \sum_{i=0}^{\ell-1}(-1)^i a^{\ell-1-i}(a^{\dag})^{i}\right) \left(\sum_{i=0}^{\ell-1}(-1)^i a^{i}(a^{\dag})^{\ell-1-i} \right) + \alpha_1(a,a^{\dag})^\dag \alpha_1(a,a^{\dag}),
  \end{align*}
  notice that for each monomial $a^{\ell-1-i}(a^{\dag})^{i}$ in the first sum, there is exactly one monomial in the second sum,
  namely $a^{i}(a^{\dag})^{\ell-1-i}$, such that their product is equal to $a^{\ell-1}(a^{\dag})^{\ell-1}$ plus some lower degree terms.
  Consequently,
  \[
    \alpha(a,a^{\dag})^{\dag} \alpha(a,a^{\dag}) = \frac{\ell \Delta^2}{4 g^2} a^{\ell-1} (a^\dag)^{\ell-1} + \alpha_{2}(a,a^{\dag}),
  \]
  where the degree of the polynomial $\alpha_{2}(a,a^{\dag})$ is at most $2\ell-2$ but does not contain the monomial $a^{\ell-1}(a^\dag)^{\ell-1}$.
  Similarly, the product $\gamma(a,a^{\dag})^{\dag} \gamma(a,a^{\dag})$ is equal to
  \begin{align*}
    &\left( (-1)^{\ell} \left( a^\ell + \ell g a^{\ell-1} \right) + \gamma_1(a,a^{\dag})^{\dag} \right) \left( (-1)^{\ell} \left( (a^{\dag})^\ell + \ell g (a^{\dag})^{\ell-1} \right) + \gamma_1(a,a^{\dag}) \right) \\
    &\qquad = a^{\ell} (a^{\dag})^\ell + \ell g (a^{\ell-1}(a^{\dag})^{\ell} + a^{\ell}(a^{\dag})^{\ell-1} ) + \ell^2 g^2 a^{\ell-1} (a^{\dag})^{\ell-1} + \gamma_2(a,a^\dag),
  \end{align*}
  where the degree of the polynomial $\gamma_2(a,a^\dag)$ is at most $2\ell-2$. Notice that $\gamma_2(a,a^\dag)$ does not contain the monomial
  $a^{\ell-1}(a^\dag)^{\ell-1}$ since the only monomial of degree $\ell$ in $\gamma(a,a^{\dag})$ is $(a^\dag)^\ell$.
  
  Summing up, the upper-left entry component of $J^2$ is given by
  \[
    (J_{\epsilon}^2)_{1,1} = a^{\ell}(a^{\dag})^{\ell} + \ell g(a^{\ell-1}(a^{\dag})^{\ell} + a^{\ell}(a^{\dag})^{\ell-1}) + \left( \ell^2 g^2 + \frac{\ell}{4}\frac{\Delta^2}{g^2}\right) a^{\ell-1}(a^\dag)^{\ell-1} + R(a,a^{\dag}),
  \]
  where the degree of $R(a,a^{\dag})$ is at most $2\ell-2$ but it does not contain the monomial $a^{\ell-1}(a^\dag)^{\ell-1}$.
  
  Next, we notice that for any $m,n \geq 0$ and constant $c \in \C$, we have
  \begin{align} \label{eq:prod2}
    a^n (a^{\dag})^m \left( a a^{\dag} + g(a+ a^{\dag}) + c \right) &= a^{n+1}(a^{\dag})^{m+1} + g a^{n}(a^{\dag})^{m+1}  + g a^{n+1} (a^{\dag})^{m}  + (c - k) a^{m}(a^\dag)^{n} - m a^{n} (a^{\dag})^{m-1}.
  \end{align} 
  Notice that the upper-left entry of $\bm{M} \HRabi{\epsilon}$ is given by
  \[
    \bm{M}_{1,1} (a a^{\dag} + g(a+a^{\dag}) + 1-\epsilon) + \bm{M}_{1,2} \Delta,
  \]
  by \eqref{eq:prod2}, it is clear that there are no $\bm{M}_{1,1}$ and $\bm{M}_{1,2}$ in $\C[a,a^{\dag}]$ such that
  \[
    \bm{M}_{1,1} (a a^{\dag} + g(a+a^{\dag}) + 1-\epsilon) + \bm{M}_{1,2} \Delta = a^{\ell}(a^{\dag})^{\ell} + \ell g(a^{\ell-1}(a^{\dag})^{\ell} + a^{\ell}(a^{\dag})^{\ell-1}) + \left( \ell^2 g^2 + \frac{\ell}{4}\frac{\Delta^2}{g^2}\right) a^{\ell-1}(a^\dag)^{\ell-1} + R(a,a^{\dag}).
  \]
  Indeed, suppose that $\bm{M}_{1,1} = a^{\ell-1}(a^\dag)^{\ell-1}$, then there is no choice of $\bm{M}_{1,2} \in \C[a,a^{\dag}]$ such that the coefficients
  of $a^{\ell-1}(a^{\dag})^{\ell}$,$a^{\ell}(a^{\dag})^{\ell-1}$ and $\frac{\ell}{4}\frac{\Delta^2}{g^2} a^{\ell-1}(a^\dag)^{\ell-1}$ in both sides of the equation above are
  satisfied simultaneously.

  Similarly, if $\bm{M}_{1,2} = \frac1{\Delta} a^{\ell}(a^\dag)^{\ell}$ no choice of $\bm{M}_{1,1} \in \C[a,a^{\dag}]$ satisfies the equation above.
  This proves that $J^2$ is not (left or right) divisible  by $\HRabi{\epsilon}$, completing the proof.
\end{proof}

\begin{rem} 

  In Example \ref{ex:J2}, for $\epsilon = \frac12$ and $\epsilon=1$ the constant term $p_{\epsilon}(0;g,\Delta)$ of the polynomial $p_{\epsilon}(x;g,\Delta)$ is given by
  a polynomial in $g$ and $\Delta$ with positive coefficients, thus $p_{\epsilon}(0;g,\Delta)>0$ for $g,\Delta>0$.
  
  This is not the case in general, for instance, for $\epsilon = \frac32$, we have 
  \[
    p_{\epsilon}(0;g,0) = 8\left(2 g^2-1\right) \left(2 g^2+1\right) \left(2 g^2+3\right).
  \]
  It is clear that the polynomial $p_{\epsilon}(0;g,\Delta)$ may take negative values or zero for particular values of $g,\Delta>0$.
  
  It also may be interesting to study the algebraic curves 
  \[
    C^{(\alpha)}_{\frac{\ell}2} = \{ (g,\Delta) \in \R^2 \, | \, p_{\frac{\ell}2}(\alpha;g,\Delta) = 0 \}
  \]
  with $\alpha\in \R$, and their structure (singularities, $L$-functions or congruent zeta function of the non-singular curve over finite fields).
  In Figure \ref{fig:Ccurves} we show the first non-trivial cases of the curves $C^{(0)}_{\frac{\ell}2}$.
  
  Note that the plane-curves given by $p_{\frac{3}2}(\alpha;g,\Delta) = 0$ define a one-parameter family of elliptic curve (rational coefficients)
  with respect to the variables $(2g)^2, \,\Delta^2$ with parameter $\alpha$. 

  \begin{figure}[h!]
    \centering
    \subfloat[$\ell=3$]{
      \includegraphics[height=4cm]{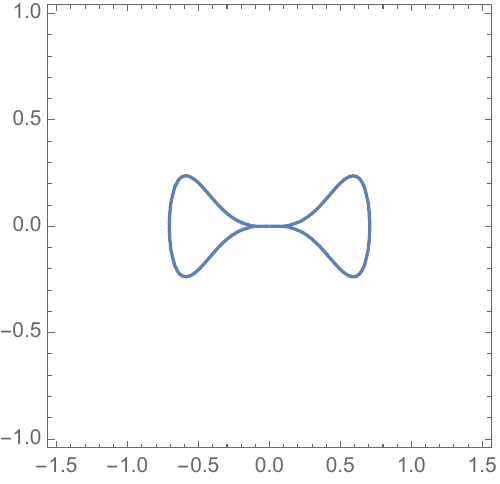}}
    ~
    \subfloat[$\ell=4$]{
      \includegraphics[height=4cm]{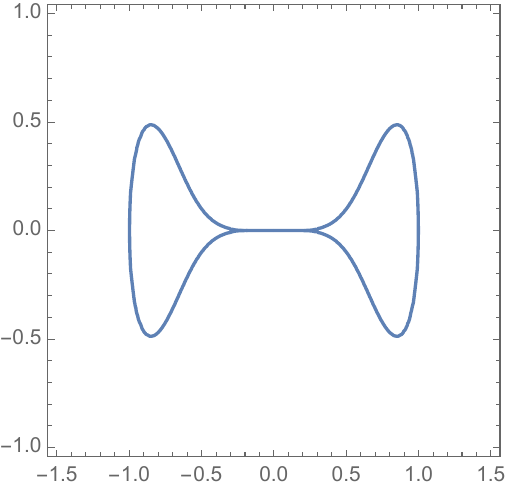}}
    ~
    \subfloat[$\ell=5$]{
      \includegraphics[height=4cm]{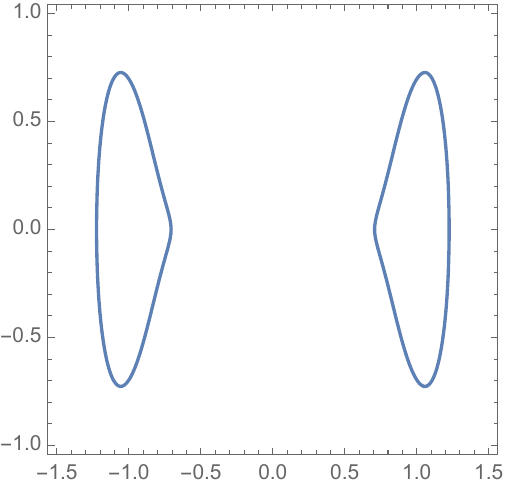}} \\
    \caption{Curves $C^{(0)}_{\frac{\ell}2}$ for $\ell=3,4,5$.}
    \label{fig:Ccurves}
\end{figure}

\end{rem}

\section*{Acknowledgements}
This work was partially supported by Grant-in-Aid for Scientific Research (C) No.16K05063 and is by No.20K03560, JSPS,
JST CREST Grant Number JPMJCR14D6, Japan, and by the Deutsche Forschungsgemeinschaft (DFG, German Research Foundation) under Grant No.439943572. 


\begin{flushleft}

\bigskip

 Cid Reyes-Bustos \par
 Department of Mathematical and Computing Science, School of Computing, \par
 Tokyo Institute of Technology \par
 2 Chome-12-1 Ookayama, Meguro, Tokyo 152-8552 JAPAN \par\par
 \texttt{reyes@c.titech.ac.jp}

 \bigskip

 Daniel Braak \par
 Department of Physics, Augsburg University, \par
 Universit\"atsstr. 1, 86159 Augsburg, GERMANY \par
 \texttt{daniel.braak@physik.uni-augsburg.de}
 
 \bigskip

 Masato Wakayama \par
 Department of Mathematics, School of Science, \par
 Tokyo University of Science \par
 1-3 Kagurazaka, Shinjyuku-ku, Tokyo 162-8601 JAPAN \par\par
 \texttt{wakayama@rs.tus.ac.jp}

\end{flushleft}

\end{document}